\documentclass[twocolumn,
showpacs,preprintnumbers,amsmath,amssymb]{revtex4}

\usepackage{amssymb,amsmath,amsthm}
\usepackage{graphicx}
\newtheorem{Thm}{Theorem}

\newtheorem{Lem}{Lemma}

\theoremstyle{definition}

\newcommand{\bra}[1]{{\left\langle #1 \right|}}
\newcommand{\ket}[1]{{\left| #1 \right\rangle}}

\newcommand{\C}{\mbox{$\mathbb C$}}

\newcommand{\T}{\mbox{$\mathrm{tr}$}}

\begin{document}
\title{Tsallis entropy and general polygamy of multi-party quantum entanglement in arbitrary dimensions}
\author{Jeong San Kim}
\email{freddie1@khu.ac.kr} \affiliation{
 Department of Applied Mathematics and Institute of Natural Sciences, Kyung Hee University, Yongin-si, Gyeonggi-do 446-701, Korea
}
\date{\today}

\begin{abstract}
We establish a unified view to the polygamy of multi-party quantum entanglement in arbitrary dimensions.
Using quantum Tsallis-$q$ entropy, we provide a one-parameter class of polygamy inequalities
of multi-party quantum entanglement. This class of polygamy inequalities
reduces to the known polygamy inequalities based on tangle and
entanglement of assistance for a selective choice of the parameter $q$.
We further provide one-parameter generalizations of various quantum correlations
based on Tsallis-$q$ entropy. By investigating the properties of the generalized quantum correlations,
we provide a sufficient condition, on which the Tsallis-$q$ polygamy inequalities hold in multi-party
quantum systems of arbitrary dimensions.
\end{abstract}

\maketitle

\section{Introduction}

Quantum entanglement is a quintessential manifestation of quantum mechanics
revealing the fundamental insights into the nature of quantum correlations.
One distinct property of quantum entanglement from other classical correlations
is its limited shareability in multi-party quantum systems, known as the
{\em monogamy of entanglement}(MoE)~\cite{T04, KGS}.

MoE was characterized in a quantitative way as an inequality;
for a given three-party quantum state $\rho_{ABC}$ with reduced density matrices $\rho_{AB}=\T_C \rho_{ABC}$ and
$\rho_{AC}=\T_B \rho_{ABC}$, and a bipartite entanglement measure $E$, {\em monogamy inequality} leads
\begin{align}
E\left(\rho_{A|BC}\right)\geq E\left(\rho_{A|B}\right)+E\left(\rho_{A|C}\right)
\label{MoE}
\end{align}
where $E\left(\rho_{A|BC}\right)$ is the bipartite entanglement between subsystems $A$ and $BC$.
Monogamy inequality shows the mutually exclusive relation of the bipartite entanglement
between $A$ and each of $B$ and $C$(measured by $E\left(\rho_{A|B}\right)$ and $E\left(\rho_{A|C}\right)$, respectively),
so that their summation cannot exceeds the total entanglement between $A$ and $BC$(measured by $E\left(\rho_{A|BC}\right)$ ).

Monogamy inequality was first proven for three-qubit systems using {\em tangle}
as the bipartite entanglement measure~\cite{ckw}, and generalized into multi-qubit systems in terms of various
entanglement measures~\cite{ov, KSRenyi, KimT, KSU}. For a general monogamy
inequality of multi-party quantum entanglement in arbitrary dimension, it was shown that
squashed entanglement~\cite{CW04} is a faithful entanglement measure~\cite{BCY10}, which also shows
a general monogamy inequality~\cite{KW}.

Whereas MoE is about the limited shareability of bipartite
entanglement in multi-party quantum systems, the {\em assisted entanglement},
which is a dual amount to bipartite entanglement measures, is known to have a dually monogamous (thus polygamous)
property in multi-party quantum systems. Moreover, this dually monogamous property of multi-party quantum entanglement
was also characterized as a dual monogamy inequality(thus polygamy inequality)~\cite{GMS},
\begin{align}
\tau_a\left(\rho_{A|BC}\right)\le\tau_a\left(\rho_{A|B}\right)
+\tau_a\left(\rho_{A|C}\right),
\label{PoE}
\end{align}
for a three-qubit state $\rho_{ABC}$, where $\tau_a\left(\rho_{A|BC}\right)$ is the tangle of assistance
of $\rho_{ABC}$ with respect to the bipartition between $A$ and $BC$.
Later, Inequality~(\ref{PoE}) was generalized into multi-qubit systems as well as some class of
higher-dimensional quantum systems~\cite{GBS, KimT}. A general polygamy inequality of multi-party quantum entanglement in arbitrary
dimensional quantum systems was established using entanglement of assistance~\cite{BGK, KimGP}.

As a one-parameter generalization of von Neumann entropy, Tsallis-$q$ entropy~\cite{tsallis, lv}
is used in many areas of quantum information theory; Tsallis entropy provides some conditions
for separability of quantum states~\cite{ar,tlb,rc}, and it is used characterize
classical statistical correlations inherented in quantum states~\cite{rr}.
There are also discussions about using the non-extensive statistical
mechanics to describe quantum entanglement in terms of Tsallis entropy~\cite{bpcp}.

Tsallis entropy also plays an important role in quantum entanglement theory.
For all parameters $q > 0$, Tsallis-$q$ entropy is a concave function on the set of density matrices,
which assures the property of {\em entanglement monotone}~\cite{vidal}.
In other words, Tsallis entropy can be used to construct a faithful entanglement measure that does not increase under
{\em local quantum operations and classical communication}(LOCC).

Here, we establish a unified view to polygamy inequalities of multi-party quantum entanglement in terms of Tsallis-$q$ entropy.
Using a class of bipartite entanglement measures, {\em Tsallis-$q$ entanglement}
as well as its dual quantities {\em Tsallis-$q$ entanglement of assistance}, we provide a
one-parameter class of polygamy inequalities in multi-party quantum systems of arbitrary dimensions.

This class of polygamy inequalities is reduced to the known polygamy inequalities based on tangle and
entanglement of assistance for a selective choice of the parameter $q$.
Thus our class of polygamy inequalities provides an interpolation among
various polygamy inequalities of multi-party quantum entanglement.

We further provide one-parameter generalizations of various quantum correlations
based on Tsallis-$q$ entropy. By investigating the properties of the generalized quantum correlations,
we provide a sufficient condition, on which the Tsallis-$q$ polygamy inequality holds in multi-party
quantum systems of arbitrary dimensions. Moreover, we show that the sufficient condition we provide here is guaranteed
for the polygamy inequality based on entanglement of assistance.
Thus our results also encapsulate the known results of general polygamy inequality in a unified view in terms of
Tsallis-$q$ entropy.

This paper is organized as follows. In Sec.~\ref{Sec: Tentropy},
we recall the definition of Tsallis-$q$ entropy, and provide some generalize entropic properties in terms of
Tsallis-$q$ entropy. In Sec.~\ref{Subsec: Tentroty}, we recall the definitions of Tsallis-$q$ entanglement
as well as its dual quantity, Tsallis-$q$ enatnglement of assistance(TEoA), and we briefly review the monogamy and polygamy inequalities in multi-party quantum systems
based on generalized entropies in Sec.~\ref{Subsec: monopolyge}.
In Sec.~\ref{Subsec: unification}, we provide a unified view of general polygamy inequality of multi-party quantum entanglement using TEoA.
In Sec.~\ref{sec: q-exp}, we generalize various quantum correlations
such as Holevo quantity, one-way unlocalizable entanglement and quantum mutual information into one-parameter classes
with respect to the parameter $q$. In Sec.~\ref{sec: ccq}, we consider a classical-classical-quantum state in
four-party quantum systems, and investigate its properties related with the generalized quantum correlations in the previous section.
In Sec.~\ref{sec: gpoly}, we show some sufficient condition for the general polygamy inequality of multi-party quantum entanglement
in arbitrary dimensions using Tsallis-$q$ entropy, and we summarize our results in Sec.~\ref{Conclusion}.

\section{Tsallis-$q$ Entropy}
\label{Sec: Tentropy}

Based on the generalized logarithmic function with respect to the parameter $q$ with $q > 0,~q \ne 1$,
\begin{eqnarray}
\ln _{q} x &=&  \frac {x^{1-q}-1} {1-q},
\label{qlog}
\end{eqnarray}
Tsallis-$q$ entropy (or Tsallis entropy of order $q$) for a probability distribution $\mathbf{P}=\{p_i\}$
is defined as
\begin{align}
H_{q}\left(\mathbf{P}\right)=-\sum_{i}p_{i}^q \ln _{q}p_i = \frac{1}{1-q}\left[\sum_{i}p_{i}^q -1\right],
\label{Ctsallis}
\end{align}
which takes the {\em $q$-expectation} of the generalized logarithmic function with respect to the probability distribution~\cite{tsallis}.
As the singularity at $q=1$ in Eq.~(\ref{qlog}) is removable by its limit value, which is the natural logarithm $\ln x$,
Tsallis-$q$ entropy in Eq.~(\ref{Ctsallis}) converges to Shannon entropy when $q$ tends to $1$,
\begin{equation}
\lim_{q\rightarrow 1}H_{q}\left(\mathbf{P}\right)=-\sum_{i}p_i \ln p_i=H\left(\mathbf{P}\right).
\label{tsallishannon}
\end{equation}

By replacing the probability distribution $\mathbf{P}$ with a density matrix $\rho$,
quantum Tsallis-$q$ entropy is defined as
\begin{align}
S_{q}\left(\rho\right)=-\T \rho ^{q} \ln_{q} \rho = \frac {1-\T\left(\rho ^q\right)}{q-1}
\label{Qtsallis}
\end{align}
for $q > 0,~q \ne 1$~\cite{lv}.
Similarly, quantum Tsallis-$q$ entropy converges to von Neumann entropy
when $q$ tends to $1$,
\begin{equation}
\lim_{q\rightarrow 1}S_{q}\left(\rho\right)=-\T\rho \ln \rho=S\left(\rho\right).
\end{equation}
For these reasons, we simply denote $S_{1}\left(\rho\right)=S\left(\rho\right)$, and thus
Tsallis-$q$ entropy is a one-parameter generalization of von Neumann entropy with respect to the parameter $q$.

It is noteworthy that Tsallis-$q$ entropy is a nonextensive generalization of von Neumann entropy. Whereas von Neumann entropy has the
{\em extensivity} (or additivity) property, that is, the joint entropy of a pair of independent systems $\rho \otimes \sigma$ is equal to the
sum of the individual entropies
\begin{align}
S\left(\rho\otimes \sigma\right)=S\left(\rho\right)+S\left(\sigma\right),
\label{pseudoadd}
\end{align}
this extensivity no longer holds for Tsallis-$q$ entropy, unless $q=1$.
Instead, Tsallis-$q$ entropy has so-called {\em pseudoadditivity } relation as
\begin{align}
S_{q}\left(\rho\otimes \sigma\right)=S_{q}\left(\rho\right)+S_{q}\left(\sigma\right)+\left(1-q\right)S_{q}\left(\rho\right)S_{q}\left(\sigma\right)
\label{nonext}
\end{align}
for $q \geq 0$.

The following lemma shows that the idea of $q$-expectation naturally generalizes some entropic property in terms of
Tsallis-$q$ entropy.
\begin{Lem}(Joint entropy theorem)
For a probability distribution $\mathbf{P}=\{p_i\}$, a set of density operators $\{\rho^i_A\}$ of a system $A$
and a set of orthogonal states $\{\ket{i}_B\}$ of another system $B$, we have
\begin{align}
S_q\left(\sum_{i}p_i\rho_A^i \otimes \ket{i}_B\bra{i}\right)=\sum_{i}p_{i}^{q}S_q \left(\rho_A^i\right)+H_q\left(\mathbf{P}\right),
\label{eq: q_joint}
\end{align}
for $q \geq 0$ and $q\neq 1$.
\label{Lem: qjoint}
\end{Lem}
\begin{proof}
From the definition of quantum Tsallis-$q$ entropy in Eq.~(\ref{Qtsallis}),
\begin{align}
S_q\left(\sum_{i}p_i\rho_A^i \otimes \ket{i}_B\bra{i}\right)=& \frac{1-\T\left(\sum_{i} p_{i} \rho_A^i\otimes\ket{i}_B\bra{i}\right)^{q}}{q-1}\nonumber\\
  =&\frac {1-\T\sum_{i}p_{i}^{q} \left(\rho_A^i\right)^{q}}{q-1}\nonumber\\
  =&\frac {1-\sum_{i}p_{i}^{q}}{q-1} +\sum_{i}p_{i}^{q}\frac{1-\T\left(\rho_A^i\right)^{q}}{q-1}\nonumber\\
  =&H_q\left(\mathbf{P}\right)+\sum_{i}p_{i}^{q}S_{q}\left(\rho_A^i\right),
\end{align}
\end{proof}

In fact, we can analogously show that Eq.~(\ref{eq: q_joint}) also holds in more general cases;
for a probability distribution $\mathbf{P}=\{p_i\}$ and a set of density operators $\{\rho^i\}$ with mutually orthogonal supports, we have
\begin{align}
S_q\left(\sum_{i}p_i\rho^i\right)=\sum_{i}p_{i}^{q}S_q \left(\rho^i\right)+H_q\left(\mathbf{P}\right),
\label{eq: gjoint2}
\end{align}
for $q \geq 0$ and $q\neq 1$.

Due to the continuity of Tsallis-$q$ entropy with respect to the parameter $q$, Eq.~(\ref{eq: q_joint}) is reduced to the joint entropy theorem
in terms of Shannon and von Neumann entropy,
\begin{align}
S\left(\sum_{i}p_i\rho_A^i \otimes \ket{i}_B\bra{i}\right)=\sum_{i}p_{i}S\left(\rho_A^i\right)+H\left(\mathbf{P}\right),
\label{eq: joint}
\end{align}
for the case when $q$ tends to $1$.
\section{Tsallis entanglement and Polygamy of multi-party quantum entanglement}
\label{Sec: TQC}
\subsection{Tsallis-$q$ entanglement}
\label{Subsec: Tentroty}

For a bipartite pure state $\ket{\psi}_{AB}$ with its reduced density matrix
$\rho_A=\T _{B} \ket{\psi}_{AB}\bra{\psi}$ onto subsystem $A$,
its Tsallis-$q$ entanglement is defined as~\cite{KimT}
\begin{equation}
{\mathcal T}_{q}\left(\ket{\psi}_{A|B} \right)=S_{q}(\rho_A).
\label{TEpure}
\end{equation}
For a bipartite mixed state $\rho_{AB}$,
its Tsallis-$q$ entanglement is defined via convex-roof extension,
\begin{equation}
{\mathcal T}_{q}\left(\rho_{A|B} \right)=\min \sum_i p_i {\mathcal T}_{q}(\ket{\psi_i}_{A|B}),
\label{TEmixed}
\end{equation}
where the minimization is taken over all possible pure state
decompositions of $\rho_{AB}$,
\begin{equation}
\rho_{AB}=\sum_{i} p_i |\psi^i\rangle_{AB}\langle\psi^i|.
\label{decomp}
\end{equation}

Because Tsallis-$q$ entropy converges to von Neumann entropy when $q$ tends to 1,
we have
\begin{align}
\lim_{q\rightarrow1}{\mathcal T}_{q}\left(\rho_{A|B} \right)=E_{\rm f}\left(\rho_{A|B} \right),
\end{align}
where $E_{\rm f}(\rho_{AB})$ is the {\em entanglement of formation}(EoF)~\cite{bdsw} of $\rho_{AB}$, defined as
\begin{align}
E_{\rm f}\left(\rho_{A|B}\right)&=\min\sum_{i}p_iS(\rho_A^i)\nonumber\\
&=\min\sum_{i}p_iS(\rho_B^i)\nonumber\\
&=E_{\rm f}\left(\rho_{B|A}\right)=E_{\rm f}\left(\rho_{AB}\right).\label{eof}
\end{align}
with the minimum taken over all possible pure state
decompositions of $\rho_{AB}$ in Eq.~(\ref{decomp}), $\rho^{i}_{A}=\T_{B}|\psi^i\rangle_{AB}\langle\psi^i|$
and $\rho^{i}_{B}=\T_{A}|\psi^i\rangle_{AB}\langle\psi^i|$.
Moreover, due to the coincidence
\begin{equation}
S_q\left(\rho^i_A\right)=S_q\left(\rho^i_B\right)
\end{equation}
for each $|\psi^i\rangle_{AB}$ in Eq.~(\ref{decomp}), we have
\begin{align}
{\mathcal T}_{q}\left(\rho_{A|B}\right)&=\min\sum_{i}p_iS_q(\rho_A^i)\nonumber\\
&=\min\sum_{i}p_iS_q(\rho_B^i)\nonumber\\
&={\mathcal T}_{q}\left(\rho_{B|A}\right).
\label{Tqcomm}
\end{align}

As a dual quantity to Tsallis-$q$ entanglement,
Tsallis-$q$ entanglement of Assistance(TEoA) is defined as~\cite{KimT}
\begin{equation}
{\mathcal T}^a_{q}\left(\rho_{A|B} \right)=\max \sum_i p_i {\mathcal T}_{q}(\ket{\psi_i}_{A|B}),
\label{TEoA}
\end{equation}
where the maximum is taken over all possible pure state
decompositions of $\rho_{AB}$.
Similarly, we have
\begin{align}
\lim_{q\rightarrow1}{\mathcal T}^a_{q}\left(\rho_{A|B}
\right)=E^a\left(\rho_{A|B} \right),
\label{TsallistoEoA}
\end{align}
where $E^a(\rho_{A|B})$ is the entanglement of assistance(EoA)
of $\rho_{AB}$ defined as~\cite{cohen}
\begin{equation}
E^a(\rho_{A|B})=\max \sum_{i}p_i S(\rho^{i}_{A}). \label{eoa}
\end{equation}
with the maximization over all possible pure state
decompositions of $\rho_{AB}$.

\subsection{Monogamy and polygamy inequalities of multi-party quantum entanglement based on generalized quantum entropies}
\label{Subsec: monopolyge}

Using Tsallis-$q$ entanglement in Eq.~(\ref{TEmixed}) to quantify bipartite quantum entanglement,
the monogamy inequality in Eq.~(\ref{MoE}) was established in multi-qubit systems; for
any $n$-qubit state $\rho_{A_1A_2\cdots A_n}$ and its two-qubit reduced density matrices $\rho_{A_1A_i}$ with $i=2,\cdots, n$,
we have
\begin{equation}
{\mathcal T}_{q}\left(\rho_{A_1|A_2\cdots A_n}\right)
\geq {\mathcal T}_{q}\left(\rho_{A_1|A_2}\right)+\cdots +{\mathcal T}_{q}\left(\rho_{A_1|A_n}\right),
\label{Tqmono}
\end{equation}
for $ 2\leq q \leq 3$~\cite{KimT}. It was also shown that TEoA can be used to characterize
the polygamy of multi-qubit entanglement as
\begin{equation}
{\mathcal T}^a_{q}\left(\rho_{A_1|A_2\cdots A_n}\right)
\leq {\mathcal T}^a_{q}\left(\rho_{A_1|A_2}\right)+\cdots +{\mathcal T}^a_{q}\left(\rho_{A_1|A_n}\right),
\label{Tqpolyqubit}
\end{equation}
for $1 \leq q \leq 2$ and $3 \leq q \leq 4$~\cite{KimT}.
Recently, more generalized monogamy and polygamy inequalities of multi-qubit entanglement was proposed
in terms of Tsallis-$q$ entanglement and TEoA for selective choices of $q$~\cite{KimAoP}.

Besides Tsallis-$q$ entropy, {\em R\'enyi-$\alpha$ entropy} is another one-parameter family of entropy functions,
which contains von Neumann entropy as a special case; for a positive real number $\alpha$ and a quantum state $\rho$,
the R\'enyi-$\alpha$ entropy of $\rho$ is defined as
\begin{equation}
R_{\alpha}(\rho)=\frac{1}{1-\alpha}\log \T \rho^{\alpha}
\label{Renyi}
\end{equation}
for $\alpha \neq 1$~\cite{renyi, horo}.
Similar to the case of Tsallis-$q$ entropy, R\'enyi-$\alpha$ entropy has a singularity at $\alpha =1$.
However this singularity is removable in the sense that R\'enyi-$\alpha$ entropy converges
to von Neumann entropy when $\alpha$ tends to $1$.

As a generalization of EoF into the full spectrum of
R\'enyi-$\alpha$ entropy, {\em R\'enyi-$\alpha$ entanglement} was introduced as
\begin{equation}
E_{\alpha}\left(\ket{\psi}_{A|B} \right)=R_{\alpha}(\rho_{A}),
\label{Eq: Renyi measure pure}
\end{equation}
for a bipartite pure state $\ket{\psi}_{\rm AB}$
and
\begin{equation}
E_{\alpha}\left(\rho_{A|B} \right)=\min \sum_i p_i E_{\alpha}\left(\ket{\psi_i}_{A|B}\right),
\label{Eq: Renyi measure mixed}
\end{equation}
for a bipartite mixed state $\rho_{\rm AB}$ with the minimum over all possible pure-state
decompositions of $\rho_{A|B}=\sum_{i}p_i \ket{\psi_i}_{AB}\bra{\psi_i}$~\cite{v, KSRenyi}.

Based on R\'enyi-$\alpha$ entanglement of order $2$(R\'enyi-$2$ entanglement), a monogamy inequality was established
for multi-qubit entanglement as
\begin{equation}
E_{2}\left(\rho_{A_1|A_2\cdots A_n}\right)
\geq E_{2}\left(\rho_{A_1|A_2}\right)+\cdots +E_{2}\left(\rho_{A_1|A_n}\right),
\label{R2mono}
\end{equation}
for a multi-qubit state $\rho_{A_1A_2\cdots A_n}$ and its two-qubit reduced density matrices $\rho_{A_1A_i}$~\cite{CO}.
Later, the validity of multi-qubit R\'enyi-$\alpha$ monogamy inequality
was shown for any $\alpha \geq2$, that is,
\begin{equation}
E_{\alpha}\left(\rho_{A_1|A_2\cdots A_n}\right)
\geq E_{\alpha}\left(\rho_{A_1|A_2}\right)+\cdots +E_{\alpha}\left(\rho_{A_1|A_n}\right),
\label{Rmono}
\end{equation}
for any multi-qubit state $\rho_{A_1A_2\cdots A_n}$ and $\alpha \geq 2$~\cite{KSRenyi}.

\subsection{Unification of polygamy inequalities}
\label{Subsec: unification}

The first polygamy inequality was established in three-qubit systems~\cite{GMS};
for a three-qubit pure state $\ket{\psi}_{ABC}$,
\begin{equation}
\tau\left(\ket{\psi}_{A|BC)}\right)\le\tau^a\left(\rho_{A|B}\right)
+\tau^a\left(\rho_{A|C}\right), \label{3dual}
\end{equation}
where
\begin{align}
\tau\left(\ket{\psi}_{A|BC}\right)=4\det\rho_A
\label{tangle}
\end{align}
is the tangle of the pure state $\ket{\psi}_{ABC}$ between $A$ and $BC$, and
\begin{align}
\tau^a\left(\rho_{A|B}\right)=\max \sum_i p_i\tau\left({\ket{\psi_i}_{A|B}}\right)
\label{tanglemix}
\end{align}
is the tangle of assistance of $\rho_{AB}=\T_C\ket{\psi}_{ABC}\bra{\psi}$ with the maximum taken over all
pure-state decompositions of $\rho_{AB}$.
Later, Inequality~(\ref{3dual}) was generalized into multi-qubit
systems~\cite{GBS}
\begin{align}
\tau^a\left(\rho_{A_1|A_2\cdots A_n}\right)
\leq& \tau^a\left(\rho_{A_1|A_2}\right)+\cdots +\tau^a\left(\rho_{A_1|A_n}\right),
\label{ntpoly}
\end{align}
for an arbitrary multi-qubit mixed state $\rho_{A_1\cdots A_n}$ and its two-qubit reduced density matrices $\rho_{A_1A_i}$
with $i=2, \cdots , n$.

For polygamy inequality beyond qubits, it was shown that EoA can be used
to establish a polygamy inequality of three-party quantum systems as
\begin{align}
E^a\left(\ket{\psi}_{A|BC)}\right)\leq& E^a\left(\rho_{A|B}\right)+E^a\left(\rho_{A|C}\right)
\label{EoApoly3}
\end{align}
for any three-party pure state $\ket{\psi}_{ABC}$ of arbitrary dimensions~\cite{BGK}. A general polygamy inequality
was established by generalizing EoA polygamy inequality in (\ref{EoApoly3})
into multi-party quantum systems as
\begin{align}
E^a\left(\rho_{A_1|A_2\cdots A_n}\right)
\leq& E^a\left(\rho_{A_1|A_2}\right)+\cdots +E^a\left(\rho_{A_1|A_n}\right),
\label{EoApoly}
\end{align}
for any multi-party quantum state $\rho_{A_1A_2\cdots A_n}$ of arbitrary dimension~\cite{KimGP}.

Now, let us consider an unified view of the polygamy inequalities of multi-party entanglement in terms of Tsallis-$q$ entropy.
For any two-qubit pure state $\ket{\psi}_{AB}$ (or any bipartite state with Schmidt-rank 2)
with a Schmidt decomposition
\begin{equation}
\ket{\psi}_{AB}=\sqrt{\lambda_1}\ket{e_0}_A\otimes\ket{f_0}_B+\sqrt{\lambda_2}\ket{e_1}_A\otimes\ket{f_1}_B,
\label{schmidt2}
\end{equation}
its tangle in Eq.~(\ref{tangle}) coincides with Tsallis-$2$ entanglement up to a constant factor
\begin{align}
\tau\left(\ket{\psi}_{A|B}\right)=4\lambda_0\lambda_1=2{\mathcal T}_{2}\left(\ket{\psi}_{A|B}\right).
\label{2tangle_tsallis}
\end{align}
Thus the tangle-based polygamy inequality in (\ref{ntpoly}) can be rephrased as
\begin{align}
{\mathcal T}^a_{2}\left(\rho_{A_1|A_2\cdots A_n}\right)
\leq& {\mathcal T}^a_{2}\left(\rho_{A_1|A_2}\right)+\cdots +{\mathcal T}^a_{2}\left(\rho_{A_1|A_n}\right),
\label{T_2poly}
\end{align}
for any multi-qubit state $\rho_{A_1\cdots A_n}$.

Due to the continuity of Tsallis-$q$ entropy, the relation between TEoA and EoA in Eq.~(\ref{TsallistoEoA})
enables us to rephrase EoA-based polygamy inequality in (\ref{EoApoly}) as
\begin{align}
{\mathcal T}^a_{1}\left(\rho_{A_1|A_2\cdots A_n}\right)
\leq& {\mathcal T}^a_{1}\left(\rho_{A_1|A_2}\right)+\cdots +{\mathcal T}^a_{1}\left(\rho_{A_1|A_n}\right).
\label{T_1poly}
\end{align}
In other words, the polygamy inequalities of multi-party quantum entanglement established so far
can be considered in an unified way using Tsallis-$q$ entropy as
\begin{align}
{\mathcal T}^a_{q}\left(\rho_{A_1|A_2\cdots A_n}\right)
\leq& {\mathcal T}^a_{q}\left(\rho_{A_1|A_2}\right)+\cdots +{\mathcal T}^a_{q}\left(\rho_{A_1|A_n}\right),
\label{Tqpolyuni}
\end{align}
for selective choices of $q$.

In the following sections, we investigate some properties of quantum correlations based on Tsallis-$q$ entropy,
and provide sufficient conditions, on which the Tsallis-$q$ polygamy inequality in (\ref{Tqpolyuni}) holds.

\section{$q$-expectation and quantum correlations}
\label{sec: q-exp}
The definition of Tsallis-$q$ entropy in Eq.~(\ref{Qtsallis}) uses the concept of $q$-expectation to generalize
von-Neumann entropy into a class of entropies parameterized by $q$. Here, we further generalize some quantum correlations
based on the idea of $q$-expectation, and investigate their properties.

For a quantum state $\rho$ and its ensemble representation $\mathcal E = \{p_i, \rho_i\}$
(equivalently, a probability decomposition $\rho=\sum_{i}p_i\rho_i$), {\em Tsallis-$q$ difference}
is defined as
\begin{align}
\chi_q\left(\mathcal E\right)=S_q\left(\rho\right)-\sum_{i}p_{i}^q S_q\left(\rho_i\right),
\label{eq: q-diff}
\end{align}
which is a one-parameter generalization of the Holevo quantity,
\begin{align}
\chi\left(\mathcal E\right)=S\left(\rho\right)-\sum_{i}p_i S\left(\rho_i\right),
\label{eq: holevo}
\end{align}
for $q=1$.
Due to the the concavity of Tsallis-$q$ entropy, Tsallis-$q$ difference is always nonnegative for $q\geq1$.

Now, let us consider a bipartite quantum state $\rho_{AB}$ with its reduced density matrix $\rho_A=\T_A\rho_{AB}$.
Each rank-1 measurement $\{M_x\}$ applied on subsystem $B$ induces a probability ensemble $\mathcal E = \{p_x, \rho_A^x\}$
of $\rho_A$ where $p_x\equiv \T[(I_A\otimes M_x)\rho_{AB}]$ is the probability of the outcome $x$ and
$\rho^x_A\equiv \T_B[(I_A\otimes {M_x})\rho_{AB}]/p_x$ is the state
of system $A$ when the outcome was $x$.
For $q\geq1$, we define {\em one-way unlocalizable $q$-entanglement}($q$-UE) as the minimum Tsallis-$q$ difference
\begin{equation}
\begin{split}
{\mathbf u}E_q^{\leftarrow}(\rho_{AB}) &= \min_{\mathcal E} \chi_q\left(\mathcal E\right),\\
\end{split}
\label{UEq2}
\end{equation}
where the minimum is taken over the ensemble representations $\mathcal E = \{p_x, \rho_A^x\}$ of $\rho_A$ induced by
all possible rank-1 measurements $\{M_x\}$ on subsystem $B$.

Due to the continuity of Tsallis-$q$ entropy with respect to the parameter $q$,
$q$-UE is reduced to the one-way unlocalizable entanglement
\begin{equation}
\begin{split}
{\mathbf u}E^{\leftarrow}(\rho_{AB}) &=  \min_{\mathcal E} \chi\left(\mathcal E\right),\\
\end{split}
\label{fragility2}
\end{equation}
when $q$ tends to $1$~\cite{BGK}.

The term {\em unlocalizable} arises for the following reasons.
Eq.~(\ref{UEq2}) together with Eq.~(\ref{eq: q-diff}) enable us to rewrite $q$-UE as
\begin{equation}
\begin{split}
{\mathbf u}E_q^{\leftarrow}(\rho_{AB})=  S_q(\rho_A)- \max_{\{M_x\}}\sum_x p^q_x S_q(\rho^x_A)
\end{split}
\label{UEq3}
\end{equation}
where the maximum is taken over all possible rank-1 measurements
$\{M_x\}$ applied on system $B$.

For a three-party purification
$\ket{\psi}_{ABC}$ of $\rho_{AB}$ such that
$\T_C\ket{\psi}_{ABC}\bra{\psi}=\rho_{AB}$,
we note that each rank-1 measurement $\{M_x\}$ applied on system $B$ induces a pure-state decomposition
of $\rho_{AC}=\T_B\ket{\psi}_{ABC}\bra{\psi}$
as
\begin{align}
\rho_{AC}=\sum_{x}p_x \ket{\phi^x}_{AC}\bra{\phi^x}
\label{ensemble}
\end{align}
where $p_x\equiv \T[(I_{AC}\otimes M_x)\ket{\psi}_{ABC}\bra{\psi}]$ and
$\ket{\phi^x}_{AC}\bra{\phi^x}\equiv \T_B[(I_{AC}\otimes M_x)\ket{\psi}_{ABC}\bra{\psi}]/p_x$.
Moreover, it is also straightforward to verify that each pure-state decomposition of
$\rho_{AC}=\sum_{x}p_x \ket{\phi^x}_{AC}\bra{\phi^x}$
induces a rank-1 measurement $\{M_x\}$ applied on system $B$.
Because we have
\begin{align}
\T_C \ket{\phi^x}_{AC}\bra{\phi^x}=\rho_A^x,
\label{indexx}
\end{align}
for each $x$, Eq.~(\ref{UEq3}) can be rewritten as
\begin{equation}
\begin{split}
{\mathbf u}E_q^{\leftarrow}(\rho_{AB})= {\mathcal T}_{q}\left(\ket{\psi}_{A|BC}\right)- \max\sum_x p^q_x {\mathcal T}_{q}\left(\ket{\phi^x}_{A|C}\right).
\end{split}
\label{UEq4}
\end{equation}

Here, ${\mathcal T}_{q}\left(\ket{\psi}_{A|BC}\right)=S_q\left( \rho_A\right)$ represents the amount of entanglement of the pure state
$\ket{\psi}_{ABC}$ between $A$ and $BC$ quantified by Tsallis-$q$ entanglement, and
$\max\sum_x p^q_x {\mathcal T}_{q}\left(\ket{\phi^x}_{A|C}\right)$ is the
maximum average entanglement(with respect to $q$-expectation)
that is possible to be concentrated on the subsystem $AC$ with the
assistance of $B$.
Thus ${\mathbf u}E_q^{\leftarrow}(\rho_{AB})$ is the residual entanglement that cannot be localized (therefore unlocalizable)
on $AC$ by the local measurement of $B$.

From the convexity of the function $f(x)=x^q$ for $q\geq1$ and the definition of TEoA in Eq.~(\ref{TEoA}), we have
\begin{align}
{\mathcal T}^a_{q}\left(\rho_{A|C}\right)
\geq&\max\sum_x p^q_x {\mathcal T}_{q}\left(\ket{\phi^x}_{A|C}\right),
\label{TEoAbound}
\end{align}
and this leads Eq.~(\ref{UEq4}) to
\begin{equation}
\begin{split}
{\mathbf u}E_q^{\leftarrow}(\rho_{AB})\geq  {\mathcal T}_{q}\left(\ket{\psi}_{A|BC}\right)- {\mathcal T}_{q}\left(\rho_{A|C}\right),
\end{split}
\label{UEqbound}
\end{equation}
for $q\geq1$.
Analogously, we also have
\begin{equation}
\begin{split}
{\mathbf u}E_q^{\leftarrow}(\rho_{AC})\geq  {\mathcal T}_{q}\left(\ket{\psi}_{A|BC}\right)- {\mathcal T}_{q}\left(\rho_{A|B}\right).
\end{split}
\label{UEqbound2}
\end{equation}

To end this section, we provide a one-parameter generalization of quantum mutual information using Tsallis-$q$ entropy;
for a bipartite quantum state $\rho_{AB}$ with reduced density matrices $\rho_A=\T_B \rho_{AB}$ and $\rho_B=\T_A \rho_{AB}$,
the {\em Tsallis-$q$ mutual entropy} is defined as
\begin{align}
I_q\left(\rho_{A:B}\right)=S_q\left(\rho_A\right)+S_q\left(\rho_B\right)-S_q\left(\rho_{AB}\right)
\label{eq: qmutul}
\end{align}
for $q\geq 1$.

Due to the continuity of Tsallis-$q$ entropy, the Tsallis-$q$ mutual entropy in Eq.~(\ref{eq: qmutul}) is reduced to the quantum mutual information,
\begin{align}
I\left(\rho_{A:B}\right)=S\left(\rho_A\right)+S\left(\rho_B\right)-S\left(\rho_{AB}\right),
\label{eq: mutul}
\end{align}
for the case that $q$ tends to $1$.
However, we do not use the term {\em mutual information} for Eq.~(\ref{eq: qmutul}) because a proper evidence of channel coding theorem
for information transmission has not been shown in the context of Tsallis entropy, even in classical sense.

\section{Classical-Classical-Quantum(ccq) states}
\label{sec: ccq}

In this section, we consider a four-party classical-classical-quantum(ccq) state $\Omega_{XYAB}$ whose quantum part $AB$ is obtained from a given bipartite quantum state
$\rho_{AB}$ by applying local unitary operations depending on the classical part $A$ and $B$.
We also evaluate the Tsallis-$q$ mutual entropies of $\Omega_{XYAB}$ as well as
its reduced density matrices, which will provide some sufficient condition for the general polygamy inequality of multi-party quantum entanglement in terms of TEoA.

For a two-qudit quantum state $\rho_{AB}$ in $\mathcal{H}_A \otimes \mathcal{H}_B \simeq \mathcal B\left(\C^d\otimes \C^d \right)$ and
the reduced density matrix $\rho_B=\T_{A}(\rho_{AB})$, let us consider a spectral decomposition,
\begin{align}
\rho_B=\sum_{i=0}^{d-1}\lambda_{i}\ket{e_i}_B\bra{e_i}.
\label{specrhoB}
\end{align}
Using the eigenvectors of $\rho_B$, we define two quantum channels $M_{0}$ and $M_{1}$
\begin{align}
M_{0}(\sigma)&=\sum_{i=0}^{d-1}\ket{e_i}\bra{e_i}\sigma\ket{e_i}\bra{e_i}\nonumber\\
M_{1}(\sigma)&=\sum_{i=0}^{d-1} |\tilde e_j\rangle\langle\tilde
e_j|\sigma|\tilde e_j\rangle\langle\tilde e_j|, \label{channels1}
\end{align}
acting on any quantum state $\sigma$ of subsystem $\mathcal{H}_B$,
where $\{ |\tilde e_j \rangle \}_{j} $ is the $d$-dimensional {\em Fourier basis},
\begin{equation}
|\tilde e_j \rangle = \frac{1}{\sqrt{d}}\sum_{k=0}^{d-1}
\omega_d^{jk}\ket{e_k},~j=0,\ldots ,d-1, \label{fourier}
\end{equation}
and $\omega_d = e^{\frac{2\pi i}{d}}$ is the $d$th-root of unity.
By using the generalized $d$-dimensional Pauli operators
\begin{align}
Z=&\sum_{j=0}^{d-1}\omega_d^j\ket{e_j}\bra{e_j},\nonumber\\
X=&\sum_{j=0}^{d-1}\ket{e_{j+1}}\bra{e_j}=\sum_{j=0}^{d-1} \omega_d^{-j}|\tilde
e_j \rangle \langle \tilde e_j |,\label{paulis}
\end{align}
Eqs.~(\ref{channels1}) can be rewritten as
\begin{equation}
M_0(\sigma)=\frac{1}{d}\sum_{b=0}^{d-1}Z^b\sigma Z^{-b},~
M_{1}(\sigma)=\frac{1}{d}\sum_{a=0}^{d-1}X^a\sigma X^{-a}.
\label{channels2}
\end{equation}

The channels $M_{0}$ and $M_{1}$ act on $\rho_B$ as
\begin{align}
M_0(\rho_B)=&\rho_B,~M_1(\rho_B)=\frac{1}{d}I_B,
\label{m0m1id1}
\end{align}
and
\begin{align}
M_1(M_0(\rho_B))=&M_0(M_1(\rho_B))=\frac{1}{d}I_B,
\label{m0m1id2}
\end{align}
thus the actions of the channels $M_{0}$ and $M_{1}$ on the subsystem $B$ of the bipartite state $\rho_{AB}$ are
\begin{align}
(I_A\otimes M_0)(\rho_{AB})&=\sum_{i=0}^{d-1} \sigma_A^i \otimes \lambda_i\ket{e_i}_B\bra{e_i},\nonumber\\
(I_A\otimes M_1)(\rho_{AB})&=\sum_{j=0}^{d-1} \tau_A^j \otimes
\frac{1}{d}|\tilde e_j \rangle_B \langle \tilde e_j|,
\end{align}
where $\lambda_i\sigma_A^i=\T_B [(I_A \otimes
\ket{e_i}_B\bra{e_i})\rho_{AB}]$ and $\tau_A^j/d=\T_B[(I_A
\otimes |\tilde e_j\rangle_B \langle\tilde e_j|)\rho_{AB}]$ for $i,~j \in \{0,\cdots, d-1\}$.

The ensembles of subsystem $A$ induced by the action of the channels $M_0$ and $M_1$ on subsystem $B$ are
\begin{align}
\mathcal E_0=\{\lambda_i,\sigma_A^i\}_i,~
\mathcal E_1:=\{\frac{1}{d},\tau_A^j\}_j,
\label{ensembles}
\end{align}
and their Tsallis-$q$ differences are
\begin{align}
\chi_q(\mathcal E_0)=&S_q(\rho_A)-\sum_{i=0}^{d-1}\lambda_i^q S_q(\sigma_A^i)
\label{chi0}
\end{align}
and
\begin{align}
\chi_q(\mathcal E_1)=&S_q(\rho_A)-\frac{1}{d^q}\sum_{i=0}^{d-1}S_q(\tau_A^j),
\label{chi1}
\end{align}
respectively.

Now, let us consider a four-qudit ccq-state $\Omega_{XYAB}$ in
$\mathcal{H}_X \otimes \mathcal{H}_Y\otimes\mathcal{H}_A \otimes \mathcal{H}_B$,
\begin{widetext}
\begin{equation}
  \Omega_{XYAB}:=\frac 1{d^2}\sum_{x,y=0}^{d-1}\ket{x}_X
  \bra{x}\otimes\ket{y}_Y\bra{y}\otimes(I_A\otimes X^x_BZ^y_B)\rho_{AB}(I_A\otimes
  Z^{-y}_BX^{-x}_B),
\label{XYAB}
\end{equation}
with the reduced density matrices
\begin{align}
\Omega_{XAB}=&\frac
1{d}\sum_{x=0}^{d-1}\ket{x}_X\bra{x}\otimes X^x_B
\left(\sum_{i=0}^{d-1} \sigma_A^i \otimes \lambda_i\ket{e_i}_B\bra{e_i}\right)X_B^{-x},
\label{XAB}
\end{align}
\begin{align}
\Omega_{YAB}=&\frac
1{d}\sum_{y=0}^{d-1}\ket{y}_Y\bra{y}\otimes
Z_B^y\left(\sum_{j=0}^{d-1} \tau_A^j \otimes
\frac{1}{d}|\tilde e_j \rangle_B \langle \tilde
e_j|\right)Z_B^{-y}, \label{YAB}
\end{align}
\end{widetext}
\begin{align}
\Omega_{AB}=\rho_A\otimes\frac{I_B}{d},~\Omega_{XY}=\frac{I_{XY}}{d^2},
\label{AB_XY}
\end{align}
and
\begin{align}
\Omega_{X}=\frac{I_{X}}{d},~\Omega_{Y}=\frac{I_{Y}}{d}.
\label{X_Y}
\end{align}

For the Tsallis-$q$ mutual entropies of $\Omega_{XYAB}$,
$\Omega_{XAB}$ and $\Omega_{YAB}$ in Eqs.~(\ref{XYAB}),
(\ref{XAB}) and (\ref{YAB}), we have
\begin{align}
I_q\left(\Omega_{XY:AB}\right)=&\frac{d^{1-q}-1}{1-q}+d^{1-q}S_q\left(\rho_A\right)\nonumber\\
&-d^{2(1-q)}S_q\left(\rho_{AB}\right),
\label{IqXYAB2}
\end{align}
\begin{align}
I_q\left(\Omega_{X:AB}\right)
=&\frac{d^{1-q}-1}{1-q}-d^{1-q}S_q\left(\rho_B\right)+d^{1-q}\chi_q(\mathcal E_0)
\label{IqXAB2}
\end{align}
and
\begin{align}
I_q\left(\Omega_{Y:AB}\right)=(1-d^{1-q})\frac{d^{1-q}-1}{1-q}+d^{1-q}\chi_q(\mathcal E_1),
\label{IqYAB2}
\end{align}
where the detail calculations can be found in Appendix~\ref{ap1}.

\section{General Polygamy inequality of multi-party quantum entanglement in terms of Tsallis entropy}
\label{sec: gpoly}

In this section, we provide some sufficient condition for the general polygamy inequality of multi-party quantum entanglement
in arbitrary dimensions using Tsallis-$q$ entropy. The following theorem shows that the subadditivity of Tsallis-$q$ mutual entropy for
ccq states implies the polygamy inequality of three-party quantum entanglement in terms of Tsallis-$q$ entanglement.
\begin{Thm}
For $q\geq 1$, and any three-party pure state $\ket{\psi}_{ABC}$ of arbitrary dimension,
we have
\begin{align}
{\mathcal T}_{q}\left(\ket{\psi}_{A|BC}\right)
\leq& {\mathcal T}^a_{q}\left(\rho_{A|B}\right)+{\mathcal T}^a_{q}\left(\rho_{A|C}\right),
\label{Tqpoly3}
\end{align}
conditioned on the subadditivity of Tsallis-$q$ mutual entropy for the ccq state in Eq.~(\ref{XYAB}), that is,
\begin{align}
I_q\left(\Omega_{XY:AB}\right)\geq I_q\left(\Omega_{X:AB}\right)+I_q\left(\Omega_{Y:AB}\right).
\label{subadd1}
\end{align}
\label{thm: Tqpoly3}
\end{Thm}

We note that TEoA in Eq.~(\ref{TEoA}) reduces to EoA in Eq.~(\ref{eoa}) for the case that $q=1$, where the general polygamy inequality
of multi-party entanglement in terms of EoA was shown as Inequality~(\ref{EoApoly})~\cite{BGK}. Thus we show the theorem for $q>1$.
We also assume that, without loss of generality, $\ket{\psi}_{ABC}$ is a three-qudit state, that is,
$\ket{\psi}_{ABC} \in \left(\C^{d}\right)^{\otimes 3}$, otherwise, we can always consider an imbedded image of $\ket{\psi}_{ABC}$
into a higher dimensional quantum system having the same dimensions of subsystems.

\begin{proof}
For the reduced density matrices $\rho_{AB}=\T_C \ket{\psi}_{ABC}\bra{\psi}$ of $\ket{\psi}_{ABC}$ on subsystem $AB$,
let us consider the ccq state in Eq.~(\ref{XYAB}).
From Eqs.~(\ref{IqXAB2}), (\ref{IqYAB2}) and (\ref{IqXYAB2}), we can rewrite Inequality~(\ref{subadd1}) as
\begin{align}
\chi_q(\mathcal E_0)+\chi_q(\mathcal E_1)\leq& S_q\left(\rho_A\right)+S_q\left(\rho_B\right)\nonumber\\
&-d^{1-q}S_q\left(\rho_{AB}\right)+\frac{\left(d^{1-q}-1\right)^2}{d^{1-q}(1-q)}.
\label{ensemine1}
\end{align}

Because $\chi_q(\mathcal E_0)$ and $\chi_q(\mathcal E_1)$ of
Eqs.~(\ref{chi0}) and (\ref{chi1}) can be obtained, respectively, from
$\rho_{AB}$ by rank-1 measurements $\{\ket{e_i}_B\bra{e_i} \}_i$
and $\{ |\tilde e_j \rangle_B \langle \tilde e_j| \}_j$ of
subsystem $B$, the rank-1 measurement
\begin{equation}
{\bf Q}_B:= \{ \frac{\ket{e_i}_B\bra{e_i}}{2}, \frac{|\tilde e_j
\rangle_B \langle \tilde e_j|}{2}\}_{i,j}, \label{povm}
\end{equation}
of subsystem $B$ provides an upperbound of $q$-UE in Eq.~(\ref{UEq2}) as
\begin{align}
{\mathbf u}E_q^{\leftarrow}(\rho_{AB})\leq \frac{\chi_q(\mathcal E_0)+\chi_q(\mathcal E_1)}{2}.
\label{tuupper1}
\end{align}
Thus, together with Inequality (\ref{ensemine1}), we have
\begin{align}
{\mathbf u}E_q^{\leftarrow}(\rho_{AB})\leq& \frac{1}{2}[S_q\left(\rho_A\right)+S_q\left(\rho_B\right)\nonumber\\
&-d^{1-q}S_q\left(\rho_{AB}\right)+\frac{\left(d^{1-q}-1\right)^2}{d^{1-q}(1-q)}].
\label{upperUEAB}
\end{align}
Moreover, we also analogously have
\begin{align}
{\mathbf u}E_q^{\leftarrow}(\rho_{AC})\leq& \frac{1}{2}[S_q\left(\rho_A\right)+S_q\left(\rho_C\right)\nonumber\\
&-d^{1-q}S_q\left(\rho_{AC}\right)+\frac{\left(d^{1-q}-1\right)^2}{d^{1-q}(1-q)}],
\label{upperUEAC}
\end{align}
for the reduced density matrix $\rho_{AC}=\T_B \ket{\psi}_{ABC}\bra{\psi}$ on subsystem $AC$.

As $S_q\left(\rho_{AB}\right)=S_q\left(\rho_{C}\right)$
and $S_q\left(\rho_{AC}\right)=S_q\left(\rho_{B}\right)$ for the three-party pure state $\ket{\psi}_{ABC}$,
Inequalities (\ref{UEqbound}) and (\ref{UEqbound2}) together with Inequalities (\ref{upperUEAB})
and (\ref{upperUEAC}) lead us to
\begin{align}
{\mathcal T}_{q}\left(\rho_{A|B}\right)+{\mathcal T}_{q}\left(\rho_{A|C}\right)\geq& 2S_q\left(\rho_A\right)\nonumber\\
&-{\mathbf u}E_q^{\leftarrow}(\rho_{AB})-{\mathbf u}E_q^{\leftarrow}(\rho_{AC})\nonumber\\
\geq& S_q\left(\rho_A\right)+\frac{\Xi_B +\Xi_C}{2}
\label{bdtog1}
\end{align}
where
\begin{align}
\Xi_B=\frac{d^{q-1}-1}{d^{q-1}}\left[\frac{d^{q-1}-1}{q-1}-S_q\left(\rho_B \right)\right]
\label{XiB}
\end{align}
and
\begin{align}
\Xi_C=\frac{d^{q-1}-1}{d^{q-1}}\left[\frac{d^{q-1}-1}{q-1}-S_q\left(\rho_C \right)\right].
\label{XiC}
\end{align}

For $q>1$, the factor $\frac{d^{q-1}-1}{d^{q-1}}$ in Eqs.~(\ref{XiB}) and (\ref{XiC}) is nonnegative.
Moreover, due to the fact that Tsallis-$q$ entropy attains its maximum value for the maximally mixed state $\frac{I_B}{d}$, we have
\begin{align}
S_q\left(\rho_B \right)\leq S_q\left(\frac{I_B}{d} \right)
=&\frac{1-d^{1-q}}{q-1}\nonumber\\
=&\frac{d^{q-1}-1}{d^{q-1}(q-1)}\nonumber\\
\leq&\frac{d^{q-1}-1}{q-1},
\label{XiBbd}
\end{align}
for $q>1$.
Similarly, we have
\begin{align}
S_q\left(\rho_C \right)\leq \frac{d^{q-1}-1}{q-1},
\label{XiCbd}
\end{align}
and thus
\begin{align}
\Xi_B\geq0,~\Xi_C\geq0
\label{XiBCbd2}
\end{align}
for $q>1$.

Inequality~(\ref{bdtog1}) together Inequalities~(\ref{XiBCbd2}), we have
\begin{align}
S_q\left(\rho_A\right) \leq {\mathcal T}_{q}\left(\rho_{A|B}\right)+{\mathcal T}_{q}\left(\rho_{A|C}\right),
\label{poly3pure}
\end{align}
which recovers Inequality~(\ref{Tqpoly3}) because ${\mathcal T}_{q}\left(\ket{\psi}_{A|BC}\right)=S_q\left(\rho_A\right)$
for three-party pure state $\ket{\psi}_{ABC}$.
\end{proof}

We note that, for $q=1$, Tsallis-$q$ mutual entropy is reduced to the quantum mutual information,
which is subadditive for ccq-states(Appendix~\ref{ap2}). Thus Theorem~\ref{thm: Tqpoly3}
guarantees the general polygamy inequality of TEoA without the subadditivity condition
(\ref{subadd1}) for $q=1$. This also recovers the results in~\cite{BGK}.

Now, we generalize the polygamy inequality of three-party quantum entanglement in Theorem~\ref{thm: Tqpoly3} into
an arbitrary multi-party quantum systems.
\begin{Thm}
For $q\geq 1$, the general polygamy inequality multi-party quantum entanglement,
\begin{align}
{\mathcal T}^a_{q}\left(\rho_{A_1|A_2\cdots A_n}\right)
\leq& {\mathcal T}^a_{q}\left(\rho_{A_1|A_2}\right)+\cdots +{\mathcal T}^a_{q}\left(\rho_{A_1|A_n}\right)
\label{Tqpolyn}
\end{align}
holds for any multi-party quantum state $\rho_{A_1A_2\cdots A_n}$ of arbitrary dimension,
conditioned on the subadditivity of Tsallis-$q$ mutual entropy for the ccq state in Eq.~(\ref{XYAB}) .
\label{thm: Tqpolyn}
\end{Thm}
\begin{proof}
We first prove the theorem for a three-party mixed state $\rho_{ABC}$, and inductively show the validity of the theorem for
an arbitrary $n$-party quantum state $\rho_{A_1A_2\cdots A_n}$.

For a three-party mixed state $\rho_{ABC}$, let us consider an optimal decomposition of $\rho_{ABC}$ for TEoA with respect to the bipartition
between $A$ and $BC$, that is,
\begin{align}
\rho_{ABC}=\sum_i p_i\ket{\psi_i}_{ABC}\bra{\psi_i},
\label{opt1}
\end{align}
with
\begin{align}
{\mathcal T}^a_{q}\left(\rho_{A|BC}\right)=\sum_i p_i {\mathcal T}_{q}\left(\ket{\psi_i}_{A|BC}\right).
\label{optTEoA1}
\end{align}

From Theorem~\ref{thm: Tqpoly3}, each $\ket{\psi_i}_{ABC}$ in Eq.~(\ref{optTEoA1}) satisfies the polygamy inequality,
\begin{align}
{\mathcal T}_{q}\left(\ket{\psi_i}_{A|BC}\right)\leq {\mathcal T}^a_{q}\left(\rho^i_{A|B}\right)+{\mathcal T}^a_{q}\left(\rho^i_{A|C}\right)
\label{polyi}
\end{align}
with $\rho^i_{AB}=\T_C \ket{\psi_i}_{ABC}\bra{\psi_i}$ and $\rho^i_{AC}=\T_B \ket{\psi_i}_{ABC}\bra{\psi_i}$,
therefore, together with Eq.~(\ref{optTEoA1}), we have
\begin{align}
{\mathcal T}^a_{q}\left(\rho_{A|BC}\right)\leq&\sum_i p_i{\mathcal T}^a_{q}\left(\rho^i_{A|B}\right)+\sum_i p_i{\mathcal T}^a_{q}\left(\rho^i_{A|C}\right)\nonumber\\
\leq&{\mathcal T}^a_{q}\left(\rho_{A|B}\right)+{\mathcal T}^a_{q}\left(\rho_{A|C}\right)
\label{poly3bd1}
\end{align}
where the second inequality is from the definition of TEoA.

Now let us assume Inequality~(\ref{poly3bd1}) is true for and $(n-1)$-party quantum state, and consider
an $n$-party quantum state $\rho_{A_1A_2\cdots A_n}$.
By considering $\rho_{A_1A_2\cdots A_n}$ as a three-party state with respect to the partition $A_1$, $A_2$ and $A_3\cdots A_n$,
Inequality~(\ref{poly3bd1}) leads us to
\begin{align}
{\mathcal T}^a_{q}\left(\rho_{A_1|A_2\cdots A_n}\right)
\leq& {\mathcal T}^a_{q}\left(\rho_{A_1|A_2}\right)+{\mathcal T}^a_{q}\left(\rho_{A_1|A_3\cdots A_n}\right), \label{polymixed1}
\end{align}
where $\rho_{A_1A_2}=\T_{A_3\cdots A_n}\rho_{A_1A_2\cdots A_n}$,
$\rho_{A_1A_3\cdots A_n}=\T_{A_2}\rho_{A_1A_2\cdots A_n}$, and
${\mathcal T}^a_{q}\left(\rho_{A_1|A_3\cdots A_n}\right)$ is TEoA of
$\rho_{A_1A_3\cdots A_n}$ with respect to the bipartition between
$A_1$ and $A_3\cdots A_n$.

Because $\rho_{A_1A_3\cdots A_n}$ in Inequality~(\ref{polymixed1}) is a
$(n-1)$-party quantum state, the induction hypothesis assures that
\begin{align}
{\mathcal T}^a_{q}\left(\rho_{A_1|A_3\cdots A_n}\right) \leq
{\mathcal T}^a_{q}\left(\rho_{A_1|A_3}\right)+\cdots + {\mathcal T}^a_{q}\left(\rho_{A_1|A_n}\right).
\label{n-1poly}
\end{align}
Thus Inequalities~(\ref{polymixed1}) and (\ref{n-1poly}) imply the polygamy
inequality of multi-party entanglement in terms of TEoA in (\ref{Tqpolyn}).
\end{proof}

Due to the relation between TEoA and EoA in Eq.~(\ref{TsallistoEoA}), Tsalli-$q$
polygamy inequality in (\ref{Tqpolyn}) is reduced to EoA-based polygamy inequality in (\ref{EoApoly}) for $q=1$.
As the quantum mutual information is subadditive for
ccq-states (Appendix~\ref{ap2}), Theorem~\ref{thm: Tqpolyn} is true without the subadditivity condition for $q=1$,
which encapsulates the results in \cite{KimGP}.


\section{Conclusion}
\label{Conclusion}

We have established a unified view to polygamy inequalities of multi-party quantum entanglement in arbitrary dimensions
using Tsallis-$q$ entropy. We have provided a one-parameter class of polygamy inequalities in multi-party quantum systems in terms of TEoA, which
provides an interpolation among various polygamy inequalities of multi-party quantum entanglement.

We have further provided one-parameter generalizations of Holevo quantity, UE and quantum mutual information.
By investigating the properties of the generalized quantum correlations related with four-party ccq-states,
we have provided a sufficient condition, on which the Tsallis-$q$ polygamy inequality holds in multi-party
quantum systems of arbitrary dimensions. We have also shown that the sufficient condition is guaranteed
for $q=1$, which is the case that Tsallis-$q$ polygamy inequality is reduced to the general polygamy inequality based on EoA.
Thus our results encapsulate the known results of EoA-based general polygamy inequality in a unified view in terms of Tsallis-$q$ entropy.

Based on the concept of $q$-expectation, our results provide one-parameter classes of various quantum correlations as well as their properties,
which are useful methods in establishing general polygamy of multi-party entanglement in arbitrary dimensions.
Noting the importance of the study on multi-party quantum entanglement,
especially in higher-dimensional systems more than qubits, our result can provide a rich reference for future
work to understand the nature of multi-party quantum entanglement.

\section*{Acknowledgments}
This research was supported by Basic Science Research Program through the National Research Foundation of Korea(NRF)
funded by the Ministry of Education, Science and Technology(NRF-2014R1A1A2056678).


\clearpage
\setcounter{equation}{0}
\begin{widetext}
\setcounter{page}{1}
\appendix

\section{Tsallis-$q$ mutual entropy of ccq-states}
\label{ap1}
Here, we provide the detail calculation of the Tsallis-$q$ mutual entropies of the ccq-state
$\Omega_{XAB}$ in Eq.~(\ref{IqXYAB2}) as well as the reduced density matrices $\Omega_{XAB}$ and $\Omega_{YAB}$
in Eqs.~(\ref{IqXAB2}) and (\ref{IqYAB2}).
Let us first consider $\Omega_{XAB}$.
From the definition of Tsallis-$q$ mutual entropy in
Eq.~(\ref{eq: qmutul}), we have
\begin{equation}
I_q\left(\Omega_{X:AB}\right)=S_q\left(\Omega_X\right)+S_q\left(\Omega_{AB}\right)-S_q\left(\Omega_{XAB}\right),
\label{IqXAB1}
\end{equation}
where Eq.~(\ref{X_Y}) implies that
\begin{equation}
S_q\left(\Omega_X\right)=S_q\left(\frac{I_{X}}{d}\right)=\frac{d^{1-q}-1}{1-q}.
\label{SqX}
\end{equation}

From Eq.~(\ref{AB_XY}), we also have
\begin{eqnarray}
S_q\left(\Omega_{AB}\right)&=&S_q\left(\rho_A\otimes\frac{I_B}{d}\right)\nonumber\\
&=&S_q\left(\rho_A\right)+S_q\left(\frac{I_B}{d}\right)+(1-q)S_q\left(\rho_A\right)S_q\left(\frac{I_B}{d}\right)\nonumber\\
&=&\frac{d^{1-q}-1}{1-q}+d^{1-q}S_q\left(\rho_A\right)
\label{SqAB}
\end{eqnarray}
where the second equality is due to the pseudoadditivity of Tsallis-$q$ entropy in Eq.~(\ref{pseudoadd}).
For $S_q\left(\Omega_{XAB}\right)$,
the joint entropy theorem in Lemma~\ref{Lem: qjoint} implies that
\begin{eqnarray}
S_q\left(\Omega_{XAB}\right)&=&H_q\left(\mathbf{I_d}\right)+\sum_{x=0}^{d-1}\frac{1}{d^q}S_q\left(\sum_{i}\sigma_A^i \otimes \lambda_i\ket{e_i}_B\bra{e_i}\right)\nonumber\\
&=&H_q\left(\mathbf{I_d}\right)+d^{1-q}\left[H_q\left(\mathbf{\Lambda}\right)+\sum_i \lambda_i^q S_q\left(\sigma_A^i\right)\right]
\label{SqXABj}
\end{eqnarray}
where $\mathbf{I_d}=\{1/d, \cdots, 1/d\}$ is the uniform probability distribution and $\mathbf{\Lambda}=\{\lambda_i\}_i$ is the spectrum of $\rho_B$.

Due to the relation
\begin{equation}
H_q\left(\mathbf{I_d}\right)=\frac{d^{1-q}-1}{1-q}=S_q\left(\frac{I_B}{d}\right)
\label{tcq1}
\end{equation}
and
\begin{equation}
H_q\left(\mathbf{\Lambda}\right)=S_q\left(\rho_B\right),
\label{tcq2}
\end{equation}
Eqs.~(\ref{SqX}), (\ref{SqAB}) and (\ref{SqXABj}) lead us to
\begin{eqnarray}
I_q\left(\Omega_{X:AB}\right)&=&\frac{d^{1-q}-1}{1-q}+d^{1-q}S_q\left(\rho_A\right)
-d^{1-q}\left[H_q\left(\mathbf{\Lambda}\right)+\sum_i \lambda_i^q S_q\left(\sigma_A^i\right)\right]\nonumber\\
&=&\frac{d^{1-q}-1}{1-q}-d^{1-q}S_q\left(\rho_B\right)+d^{1-q}\chi_q(\mathcal E_0),
\label{IqXAB2a}
\end{eqnarray}
where $\chi_q(\mathcal E_0)$ is the Tsallis-$q$ difference of the induced ensemble $\mathcal E_0$ in Eq.~(\ref{chi0}).

For the Tsallis-$q$ mutual entropy of $\Omega_{YAB}$, we have
\begin{eqnarray}
I_q\left(\Omega_{Y:AB}\right)&=&S_q\left(\Omega_Y\right)+S_q\left(\Omega_{AB}\right)-S_q\left(\Omega_{YAB}\right)\nonumber\\
&=&2\frac{d^{1-q}-1}{1-q}+d^{1-q}S_q\left(\rho_A\right)-S_q\left(\Omega_{YAB}\right).
\label{IqYAB1}
\end{eqnarray}
Because
\begin{eqnarray}
S_q\left(\Omega_{YAB}\right)&=&H_q\left(\mathbf{I_d}\right)+\sum_{y=0}^{d-1}\frac{1}{d^q}S_q\left(\sum_{j}\tau_A^i \otimes \frac{1}{d}|\tilde e_j \rangle_B \langle \tilde
e_j|\right)\nonumber\\
&=&(1+d^{1-q})H_q\left(\mathbf{I_d}\right)+d^{1-q}\sum_j\frac{1}{d^q}S_q\left(\tau_A^j\right),
\label{SqYABj}
\end{eqnarray}
where the second equality is due to the joint entropy theorem in Lemma~\ref{Lem: qjoint}, Eqs.~(\ref{IqYAB1}) and (\ref{SqYABj}) lead us to
\begin{equation}
I_q\left(\Omega_{Y:AB}\right)=(1-d^{1-q})\frac{d^{1-q}-1}{1-q}+d^{1-q}\chi_q(\mathcal E_1),
\label{IqYAB2a}
\end{equation}
where $\chi_q(\mathcal E_1)$ is the Tsallis-$q$ difference of the induced ensemble $\mathcal E_1$ in Eq.~(\ref{chi1}).

For the Tsallis-$q$ mutual entropy of $\Omega_{XYAB}$, we have
\begin{equation}
I_q\left(\Omega_{XY:AB}\right)=S_q\left(\Omega_{XY}\right)+S_q\left(\Omega_{AB}\right)-S_q\left(\Omega_{XYAB}\right),
\label{IqXYAB1}
\end{equation}
where Eqs.~(\ref{AB_XY}) imply that
\begin{equation}
S_q\left(\Omega_{XY}\right)=S_q\left(\frac{I_{XY}}{d^2}\right)=H_q\left(\mathbf{I_{d^2}}\right),
\label{SqXY}
\end{equation}
for the uniform probability distribution $\mathbf{I_{d^2}}=\{1/d^2, \cdots, 1/d^2\}$.
Moreover, from the the joint entropy theorem in Lemma~\ref{Lem: qjoint}, we have
\begin{equation}
S_q\left(\Omega_{XYAB}\right)=H_q\left(\mathbf{I_{d^2}}\right)+\sum_{x,y}\frac{1}{d^{2q}}S_q\left(\rho_{AB}\right),
\label{SqXYAB1}
\end{equation}
therefore Eq.~(\ref{IqXYAB1}) together with Eqs.~(\ref{SqAB}), (\ref{SqXY}) and (\ref{SqXYAB1}) lead us to
\begin{equation}
I_q\left(\Omega_{XY:AB}\right)=\frac{d^{1-q}-1}{1-q}+d^{1-q}S_q\left(\rho_A\right)-d^{2(1-q)}S_q\left(\rho_{AB}\right).
\label{IqXYAB2a}
\end{equation}

\section{Subadditivity of quantum mutual information for ccq-states}
\label{ap2}
Here we provide a detail proof that the quantum mutual information in Eq.~(\ref{eq: mutul})
is subadditive for general ccq-states of the form
\begin{equation}
  \Gamma_{XYAB}=\frac 1{d^2}\sum_{x,y=0}^{d-1}\ket{x}_X
  \bra{x}\otimes\ket{y}_Y\bra{y}\otimes\sigma^{xy}_{AB},
\label{gXYAB}
\end{equation}
which has the ccq-state in Eq.~(\ref{XYAB}) as a special case.
Then the subadditivity of quantum mutual information for $\Gamma_{XYAB}$ in Eq.~(\ref{gXYAB})
is equivalent to the nonnegativity
\begin{equation}
I\left(\Gamma_{XY:AB}\right)-I\left(\Gamma_{X:AB}\right)-I\left(\Gamma_{Y:AB}\right)\geq 0.
\label{nonneg1}
\end{equation}

Let us first consider the mutual information
\begin{eqnarray}
I\left(\Gamma_{XY:AB}\right)=S\left(\Gamma_{XY}\right)+S\left(\Gamma_{AB}\right)-S\left(\Gamma_{XYAB}\right).
\label{0mut1234}
\end{eqnarray}
Due to the joint entropy theorem in Eq.~(\ref{eq: joint}), the von Neumann entropy of $\Gamma_{XYAB}$ is
\begin{eqnarray}
S\left(\Gamma_{XYAB}\right)=H\left(\mathbf{I_{d^2}}\right)+\sum_{x,y}\frac{1}{d^2}S\left( \sigma^{xy}\right)
=2\log d+\sum_{x,y}\frac{1}{d^2}S\left( \sigma^{xy}\right).
\label{vonXYAB}
\end{eqnarray}
for the uniform probability distribution $\mathbf{I_{d^2}}=\{1/d^2, \cdots, 1/d^2\}$.

Because the reduced density matrices
\begin{equation}
\Gamma_{XY}=\frac 1{d^2}\sum_{x,y=0}^{d-1}\ket{x}_X
\bra{x}\otimes\ket{y}_Y\bra{y}
\label{redu4}
\end{equation}
is a $d^2$-dimensional maximally mixed state, its von Neumann entropy is
\begin{equation}
S\left(\Gamma_{XY}\right)=S\left(\frac{I_{XY}}{d^2}\right)=
2\log d.
\label{vonXY}
\end{equation}
Thus, together with the reduced density matrix
\begin{equation}
\Gamma_{AB}=\frac 1{d^2}\sum_{x,y=0}^{d-1}\sigma^{xy}_{AB},
\label{redu5}
\end{equation}
Eqs.~(\ref{vonXYAB}), (\ref{vonXY}) imply
\begin{eqnarray}
I\left(\Gamma_{XY:AB}\right)&=&S\left(\Gamma_{XY}\right)+S\left(\Gamma_{AB}\right)-S\left(\Gamma_{XYAB}\right)\nonumber\\
&=&S\left(\frac 1{d^2}\sum_{x,y=0}^{d-1}\sigma^{xy}_{AB}\right)-\sum_{x,y}\frac{1}{d^2}S\left( \sigma^{xy}_{AB}\right).
\label{mut1234}
\end{eqnarray}

Similarly, for the reduced density matrices
\begin{equation}
\Gamma_{XAB}=\frac {1}{d}\sum_{x=0}^{d-1}\left(\ket{x}_X
  \bra{x}\otimes\sum_{y=0}^{d-1}\sigma^{xy}_{AB}\right)
\label{redu2}
\end{equation}
and
\begin{equation}
\Gamma_{YAB}=\frac {1}{d}\sum_{y=0}^{d-1}\left(\ket{y}_Y
  \bra{y}\otimes\sum_{x=0}^{d-1}\sigma^{xy}_{AB}\right),
\label{redu3}
\end{equation}
we have
\begin{equation}
I\left(\Gamma_{X:AB}\right)=S\left(\frac 1{d^2}\sum_{x,y=0}^{d-1}\sigma^{xy}_{AB}\right)-\sum_{x=0}^{d-1}\frac{1}{d}S\left(\sum_{y=0}^{d-1}\sigma^{xy}_{AB}\right)
\label{mut134}
\end{equation}
and
\begin{equation}
I\left(\Gamma_{Y:AB}\right)=S\left(\frac 1{d^2}\sum_{x,y=0}^{d-1}\sigma^{xy}_{AB}\right)-\sum_{y=0}^{d-1}\frac{1}{d}S\left(\sum_{x=0}^{d-1}\sigma^{xy}_{AB}\right).
\label{mut234}
\end{equation}
From Eqs.~(\ref{mut1234}), (\ref{mut134}) and (\ref{mut234}), the nonnegativity in (\ref{nonneg1}) can be rephrased as
\begin{eqnarray}
\sum_{y=0}^{d-1}\frac{1}{d}\left[S\left(\sum_{x=0}^{d-1}\frac{1}{d}\sigma^{xy}_{AB}\right)-\sum_{x=0}^{d-1}\frac{1}{d}S\left(\sigma^{xy}_{AB}\right)\right]
\geq S\left(\sum_{x,y=0}^{d-1}\frac {1}{d^2}\sigma^{xy}_{AB}\right)-\sum_{x=0}^{d-1}\frac{1}{d}S\left(\sum_{y=0}^{d-1}\frac{1}{d}\sigma^{xy}_{AB}\right).
\label{nonneg2}
\end{eqnarray}

Now, let us denote
\begin{equation}
\rho=\sum_{x,y=0}^{d-1}\frac 1{d^2}\sigma^{xy}_{AB}
\label{rhoxy}
\end{equation}
and consider a probability ensemble of $\rho$
\begin{equation}
\mathcal E_x=\{\frac{1}{d}, \rho^x\},~\rho^x=\sum_{y=0}^{d-1}\frac{1}{d}\sigma^{xy}_{AB}
\label{ensrho}
\end{equation}
for each $x$.
Then the right-hand side of Inequality~(\ref{nonneg2}) is the Holevo quantity of $\rho$ with respect to the ensemble $\mathcal E_x$,
\begin{equation}
\chi \left(\mathcal E_x\right)=S\left(\rho\right)-\sum_{x=0}^{d-1}\frac{1}{d}S\left(\rho^x\right),
\label{rhoxyhole}
\end{equation}
which also has an alternative representation
\begin{equation}
\chi \left(\mathcal E_x\right)=\sum_{x=0}^{d-1}\frac{1}{d}S\left(\rho^x \| \rho \right)
\label{chirela}
\end{equation}
in terms of the quantum relative entropy
\begin{equation}
S\left(\rho \|\sigma \right)=\T\rho\log\rho-\T\rho\log\sigma.
\label{relative}
\end{equation}

By denoting
\begin{equation}
\rho^y=\sum_{x=0}^{d-1}\frac {1}{d}\sigma^{xy}_{AB}
\label{rhoy}
\end{equation}
and considering a probability ensemble of $\rho^y$
\begin{equation}
\mathcal E_y=\{\frac{1}{d}, \sigma^{xy}_{AB}\}
\label{ensrhoy}
\end{equation}
for each $y$, a similar argument enables us to rephrase the left-hand side of Inequality~(\ref{nonneg2}) as
\begin{equation}
\sum_{y=0}^{d-1}\frac{1}{d}\left[\sum_{x=0}^{d-1}\frac{1}{d}S\left(\sigma^{xy}_{AB} \| \rho^y \right)\right].
\label{lefeq}
\end{equation}

From Inequality~(\ref{nonneg2}) together with Eqs.~(\ref{chirela}) and (\ref{lefeq}), the nonnegativity in (\ref{nonneg1})
is now equivalent to
\begin{equation}
\sum_{x, y=0}^{d-1}\frac{1}{d^2} S\left(\sigma^{xy}_{AB} \| \rho^y \right) \geq \sum_{x=0}^{d-1}\frac{1}{d}S\left(\rho^x \| \rho \right)
\label{nonneg4}
\end{equation}
which is always true due to the joint convexity of quantum relative entropy
\begin{equation}
\sum_i p_i S\left(\rho_i \| \sigma_i \right) \geq S\left( \sum_i p_i \rho_i \| \sum_i p_i \sigma_i \right),
 \label{jconvex}
\end{equation}
for quantum states $\rho_i$'s, $\sigma_i$'s and a probability distribution $\{p_i \}$.
\end{widetext}
\end{document}